\pdfoutput=1
\RequirePackage{ifpdf}
\ifpdf 
\documentclass[pdftex]{sigma}
\else
\documentclass{sigma}
\fi

\def\pt#1#2{\frac{\partial #1}{\partial #2}} 

\begin{document}

\allowdisplaybreaks

\renewcommand{\thefootnote}{$\star$}

\newcommand{\arXivNumber}{1510.00181}

\renewcommand{\PaperNumber}{012}

\FirstPageHeading

\ShortArticleName{The Hojman Construction and Hamiltonization of Nonholonomic Systems}

\ArticleName{The Hojman Construction and Hamiltonization\\ of Nonholonomic Systems\footnote{This paper is a~contribution to the Special Issue
on Analytical Mechanics and Dif\/ferential Geometry in honour of Sergio Benenti.
The full collection is available at \href{http://www.emis.de/journals/SIGMA/Benenti.html}{http://www.emis.de/journals/SIGMA/Benenti.html}}}

\Author{Ivan~A.~{BIZYAEV}~$^{\dag\ddag}$, Alexey~V.~{BORISOV}~$^{\dag\S}$ and Ivan~S.~{MAMAEV}~$^\dag$}

\AuthorNameForHeading{I.A.~Bizyaev, A.V.~Borisov and I.S.~Mamaev}

\Address{$^\dag$~Udmurt State University, 1 Universitetskaya Str., Izhevsk, 426034 Russia}
\EmailD{\href{mailto:bizaev_90@mail.ru}{bizaev\_90@mail.ru}, \href{mailto:borisov@rcd.ru}{borisov@rcd.ru}, \href{mailto:mamaev@rcd.ru}{mamaev@rcd.ru}}

\Address{$^\ddag$~St.~Petersburg State University, 1 Ulyanovskaya Str., St.~Petersburg, 198504 Russia}

\Address{$^\S$~National Research Nuclear University MEPhI, 31 Kashirskoe highway, Moscow, 115409 Russia}

\ArticleDates{Received October 05, 2015, in f\/inal form January 26, 2016; Published online January 30, 2016}

\Abstract{In this paper, using the Hojman construction, we give
examples of various Poisson brackets which dif\/fer from those which are
usually analyzed in Hamiltonian mechanics. They possess a nonmaximal rank,
and in the general case an invariant measure and Casimir functions can be
globally absent for them.}

\Keywords{Hamiltonization; Poisson bracket; Casimir functions;
invariant measure; nonholonomic hinge; Suslov problem; Chaplygin sleigh}

\Classification{37J60; 37J05}

\renewcommand{\thefootnote}{\arabic{footnote}}
\setcounter{footnote}{0}

\section{Introduction}\label{section0}

This paper is dedicated to Professor S.~Benenti on the occasion of his
60th birthday. He has made an essential contribution to the development of
various branches of mechanics and mathematics. In nonholonomic mechanics,
which is the subject of this paper, he proposed a new constructive form of
the equations of motion \cite{Ben02} and gave a new example of nonlinear
nonholonomic constraint~\cite{Ben01}.

In this paper we consider the possibility of representing some well-known
and new nonholonomic systems in Hamiltonian, more precisely, Poisson form.
We note from the outset that such a representation is usually possible for
reduced (incomplete) equations \cite{Reduction} and after rescaling time
(conformal Hamiltonianity).

Hamiltonian systems are the best-studied class of systems from the
viewpoint of the theory of integrability, stability, topological analysis,
perturbation theory (KAM theory) etc. Therefore, the investigation of the
possibility of representing the equations of motion in conformally
Hamiltonian form is an important, but poorly studied problem, which is
usually called the {\it Hamiltonization problem}.

There are various well-known obstructions to Hamiltonization, which are
examined in detail in \cite{4bbm} (see also \cite{KozlovBizyaev}). This
problem has several aspects: local, semilocal, and global. In this sense
the Hamiltonization problem is in many respects analogous to the problem
of existence of analytic integrals and a~smooth invariant measure (see
\cite{Kozlovv2}).

In nonholonomic mechanics, the reducing multiplier method of
Chaplygin~\cite{b21-3} (developed further in \cite{BBM_ch}) provides a
powerful tool for the Hamiltonization of equations of motion. However,
this method does not apply to the problems presented below due to the fact
that it is impossible to represent the equations of motion in the form of
the Chaplygin system, or the examples considered exhibit everywhere the
absence of a smooth invariant measure with density depending only on
positional variables (as required by the Chaplygin method). Other
Hamiltonization methods are examined in~\cite{b41}. We note that the
above-mentioned paper is of formal character and gives no nontrivial
mechanical examples encountered in applications.

Poisson brackets whose appearance cannot be explained by employing the
Chaplygin method were obtained in \cite{BBM02, Tsiganov3,BKM, Tsiganov1, Tsiganov2} with the help of an explicit algebraic ansatz. In this paper it
is shown that these brackets can be obtained using an algorithm that we
call the {\it Hojman construction} (its foundations were laid down already
by S.~Lie). It is based on the existence of conformal symmetry f\/ields,
i.e., symmetry f\/ields obtained after rescaling time, and was developed
further in \cite{Ho2, Hoj}.

The rank of the Poisson bracket obtained by using the Hojman construction
is, as a rule, smaller than the dimension of the phase space; for it a
smooth (or even singular) invariant measure and global Casimir functions
may be absent in the general case, and the system's behavior itself may
dif\/fer greatly from Hamiltonian behavior, moreover, there may exist limit
cycles. Therefore, the resulting Poisson brackets are of rather limited
utility. In practice, such a~Hamiltonian representation may turn
out to be useless (especially for inf\/inite-dimensional systems
\cite{Ho2}).

Indeed, the basic methods of Hamiltonian mechanics have been developed for
canonical systems. However, as shown in \cite{BBB}, the
representation found by us can be useful for the investigation of
stability problems.

Problems of Hamiltonization and various aspects of the dynamics
of nonholonomic systems are investigated in \cite{Tsiganov3, Kaz,drift, Bizz, SS,
faso,Fasso, Tsiganov1,Tsiganov2}.

\section[Tensor invariants and Hamiltonization of dynamical systems]{Tensor invariants and Hamiltonization\\ of dynamical systems}\label{section1}

\subsection{Basic notions and def\/initions}

Let a dynamical system be given on some manifold $\mathcal{M}^n$, a phase space.
\begin{gather}
\label{eq01}
\dot{\boldsymbol{x}}=\boldsymbol{v}(\boldsymbol{x}),
\end{gather}
where $\boldsymbol{x}=(x_1,\dots, x_n)$ are local coordinates, $\boldsymbol{v}=(v_1(\boldsymbol{x}),\dots, v_n(\boldsymbol{x}) )$.

The evolution of an arbitrary (smooth) function on $\mathcal{M}^n$ along the vector f\/ield \eqref{eq01} is given by the linear
dif\/ferential operator
\begin{gather*}
\dot{F}=\boldsymbol{v}(F)=\sum_{i=1}^n v_i\frac{\partial F}{\partial x_i}.
\end{gather*}

As is well known, the behavior of a dynamical system is in many respects determined by tensor invariants (conservation laws),
from which one can, as a rule, draw the main data on its dynamics.
For exact formulations see~\cite{BMzs}. We only note that three kinds of tensor invariants are encountered particularly frequently:
\begin{itemize}\itemsep=0pt
\item[--] {\it the first integral} $F(\boldsymbol{x})$ {\it for which}
\begin{gather*}
\dot{F}=\boldsymbol{v}(F)=0;
\end{gather*}
\item[--] {\it symmetry field} $\boldsymbol{u}(\boldsymbol{x})$ {\it for which}
\begin{gather*}
[\boldsymbol{v}, \boldsymbol{u}]=0;
\end{gather*}
\item[--] {\it the invariant measure} $\mu=\rho(\boldsymbol{x})dx_1 \wedge \cdots \wedge dx_n$
{\it whose density $\rho(\boldsymbol{x})$ is everywhere larger than zero and satisfies the Liouville equation}
\begin{gather}
\label{eq1_2}
{\rm div}(\rho \boldsymbol{v})=0,
\end{gather}
\end{itemize}
where $\wedge$ denotes the exterior product. In this paper, unless otherwise stated, all geometric objects (f\/irst integrals, vector f\/ields etc.) are assumed to be analytical
on $\mathcal{M}^n$.

Another important tensor invariant is the Poisson structure, which allows
the equations of motion to be represented in Hamiltonian form. In
applications, such a representation can usually be obtained only after
rescaling time as $d\tau=\mathcal{N}(\boldsymbol{x})dt$ (if the system is
a priori not Hamiltonian), i.e., the equations of motion~\eqref{eq01} are
represented in conformally Hamiltonian form
\begin{gather}
\label{eq02}
\dot{\boldsymbol{x}}=\mathcal{N}(\boldsymbol{x}) {\bf J}(\boldsymbol{x}) \frac{\partial H}{\partial \boldsymbol{x}},
\end{gather}
where $H(\boldsymbol{x})$ is a Hamiltonian that is a f\/irst integral, and ${\bf J}(\boldsymbol{x})=\|J_{ij}(\boldsymbol{x}) \|$ is a Poisson structure,
i.e., a skew-symmetric tensor f\/ield satisfying the Jacobi identity
\begin{gather}
\label{eqJ}
\sum_{l=1}^n\left(J_{lk}\pt{J_{ij}}{x_l}+J_{li}\pt{J_{jk}}{x_l}+J_{lj}\pt{J_{ki}}{x_l}\right)=0, \qquad i,j,k=1,\ldots n,
\end{gather}
and $\mathcal{N}(\boldsymbol{x})$ is a scalar function~-- a reducing
multiplier. The Poisson structure~${\bf J}(\boldsymbol{x})$ allows one to
def\/ine in a natural way the Poisson bracket of the functions~$f$ and~$g$
by the formula
\begin{gather*}
\{f(\boldsymbol{x}), g(\boldsymbol{x}) \}=\sum_{i,j=1}^nJ_{ij}\frac{\partial f }{\partial x_i}
\frac{\partial g }{\partial x_j}.
\end{gather*}

\begin{remark}
Time rescaling cannot be applied at points where the reducing multiplier $\mathcal{N}(\boldsymbol{x})$ vanishes,
hence, Hamiltonization is possible only on the set $\widetilde{\mathcal{M}}^n=\{\boldsymbol{x}\in \mathcal{M}^n \, | \, \mathcal{N}(\boldsymbol{x})\neq 0 \}$.
As a~result, the behavior of the trajectories in general (on the entire phase space~$\mathcal{M}^n$) can considerably dif\/fer from
the behavior of the trajectory of Hamiltonian systems. In particular, not only tori, but also two-dimensional integral manifolds can be arbitrary in this case,
see, for example, the Suslov problem~\cite{BKM}.
\end{remark}

In many examples the Poisson structure turns out to be
\begin{gather*}
\operatorname{rank} {\bf J}<n.
\end{gather*}
As is well known, the entire phase space (in the domain of the constant
rank~${\bf J}$) is foliated by symplectic leaves $\mathcal{O}$ of
dimension $\dim \mathcal{O}=\operatorname{rank} {\bf J}$. In the simplest
case the symplectic leaves are given as the level surfaces of a set of
global Casimir functions
\begin{gather}
\mathcal{O}_{\boldsymbol{c}}=\big\{ \boldsymbol{x} \in \mathcal{M}^n \, | \, C_1(\boldsymbol{x})=c_1,\dots, C_m(\boldsymbol{x})=c_m \big\},\nonumber\\
\sum_{j=1}^n J_{ij}(\boldsymbol{x})\frac{\partial C_k}{\partial x_j}=0, \qquad i=1,\dots, n, \qquad k=1,\dots, m=n - \operatorname{rank} {\bf J}.\label{ma01}
\end{gather}
Nevertheless, in most of the examples considered below the number of
global Casimir functions turns out to be less than~$m$. This may lead to
unusual (from the viewpoint of the standard theory of Hamiltonian systems)
behavior of trajectories of the system~\eqref{eq02}.

We note that there are three levels of analysis of the problem of the existence of tensor invariants:
\begin{itemize}\itemsep=0pt
\item[--] {\it local level} -- in a neighborhood of a nonsingular point of the vector f\/ield~\eqref{eq01},
\item[--] {\it semilocal level}~-- in a neighborhood of invariant sets of the system~\eqref{eq01}, such as
 f\/ixed points, periodic orbits, invariant tori etc.,
\item[--] {\it global level} -- on the entire phase space $\mathcal{M}^n$.
\end{itemize}

From the local point of view, obstructions to the existence of any tensor
invariants are not encountered by virtue of the rectif\/ication theorem for
vector f\/ields. Therefore, in what follows it is implied that we consider
the system from the semilocal and global points of view.

\subsection{The general Hojman construction}

Consider a conformally Hamiltonian system in the form \eqref{eq02}. From
the point of view of integrability, the case where
\begin{gather*}
\operatorname{rank} {\bf J}=2
\end{gather*}
is regarded as the simplest case. As a rule, in this case it is assumed
that the system \eqref{eq02} admits a natural (global) restriction to the
two-dimensional symplectic leaf $\mathcal{O}^2$ and reduces to a
Hamiltonian system with one degree of freedom, and its trajectories are
the level lines of the restriction of the
Hamiltonian~$H|_{\mathcal{O}^2}$.

It turns out that such a conclusion cannot be drawn in the general case,
when the properties of solutions~\eqref{ma01} are unknown to us. As
already noted, a full set of global Casimir functions can be absent in
${\bf J}$ (some of them are def\/ined only locally). As a result, the
symplectic leaf can be immersed in the phase space in a fairly complicated
(in particular, chaotic) way, and the above-described picture is not
always realized. Nevertheless, the above Poisson structure turns out to
be useful for the investigation of stability problems~\cite{Kon}.

Consider a Poisson structure ${\bf J}(\boldsymbol{x})$ of rank 2 on
$\mathcal{M}^n$ for which there is a pair of (globally def\/ined) vector
f\/ields $\boldsymbol{v}(\boldsymbol{x})$ and
$\boldsymbol{u}(\boldsymbol{x})$ such that
\begin{gather}
\label{eq06}
{\bf J}(\boldsymbol{x})=\boldsymbol{v}(\boldsymbol{x})\wedge\boldsymbol{u}(\boldsymbol{x}).
\end{gather}

\begin{remark}
Every skew-symmetric bivector f\/ield of constant rank $2$ locally admits such a~representation.
\end{remark}

According to the Darboux theorem, the Jacobi identity
\eqref{eqJ} holds only in the case where the distribution given by these vector f\/ields is integrable
\begin{gather}
\label{meq06}
[\boldsymbol{v},\boldsymbol{u}]=\mu_1(\boldsymbol{x})\boldsymbol{v} + \mu_2(\boldsymbol{x})\boldsymbol{u}.
\end{gather}
The Casimir functions are simultaneously the integrals of the f\/ields $\boldsymbol{v}$ and $\boldsymbol{u}$:
\begin{gather*}
\boldsymbol{u}(C_k)=\boldsymbol{v}(C_k)=0, \qquad k=1 ,\dots, m'.
\end{gather*}

\begin{remark}
Condition \eqref{meq06} was known already to S.~Lie; it can also be
obtained if the Jacobi identity is rewritten by means of the Schouten
bracket $[[ {\bf J}, {\bf J} ]]=0$. Using the properties of the Schouten
bracket, for the Poisson structure~\eqref{eq06} we f\/ind~\cite{Kon}:
\begin{gather*}
\mbox{}[[ {\bf J}, {\bf J} ]]=2 [\boldsymbol{v},\boldsymbol{u}]\wedge\boldsymbol{v}\wedge\boldsymbol{u}=0,
\end{gather*}
whence we obtain condition~\eqref{meq06}.
\end{remark}

It turns out that in many examples such decomposable Poisson structures allow one to Hamiltonize in a natural way the dynamical systems in question. This approach was
proposed in \cite{Ho2, Hoj}. We present the necessary results and omit the proofs, which are completely straightforward.

\begin{theorem}[Hojman~\cite{Hoj}]
\label{t01}
Suppose that the system \eqref{eq01} possesses a first integral~$H$ and a vector field $\boldsymbol{u}(\boldsymbol{x})$ such that
\begin{gather}
\label{eq04}
[\boldsymbol{v}, \boldsymbol{u}]=\mu(\boldsymbol{x})\boldsymbol{v}, \qquad
H_u(\boldsymbol{x})=\boldsymbol{u}(H)\not \equiv 0,
\end{gather}
then the initial system \eqref{eq01} can be represented in conformally Hamiltonian form
\begin{gather}
\label{eq07}
\dot{\boldsymbol{x}}=H_u^{(-1)}(\boldsymbol{x}){\bf J}^{(2)}(\boldsymbol{x}) \frac{\partial H}{\partial \boldsymbol{x}}, \qquad {\bf J}^{(2)}(\boldsymbol{x})=\boldsymbol{v}\wedge\boldsymbol{u}.
\end{gather}
\end{theorem}

\begin{remark}
The above conformally Hamiltonian representation \eqref{eq07} is not def\/ined for points $\boldsymbol{x}$
 at which $H_u(\boldsymbol{x})=0$, and the symbol $\not \equiv$ means that~$H_u(\boldsymbol{x})$ must not be equal to zero for any~$\boldsymbol{x}$. 	
\end{remark}	

It follows from conditions \eqref{eq04} that the vector f\/ield $\boldsymbol{u}(\boldsymbol{x})$ is not tangent to the level surface
$H(\boldsymbol{x})={\rm const}$ and, moreover,
\begin{gather*}
\boldsymbol{v}(H_u)=0,
\end{gather*}
hence, in this case the reducing multiplier $H_u(\boldsymbol{x})$ is the
f\/irst integral of the system. If, in addition, the system possesses
additional tensor invariants, then a natural generalization of this result
holds (see \cite{Hoj} for details).

\begin{proposition}
\label{pro1} Suppose that the system~\eqref{eq01} satisfies the conditions
of Theorem~{\rm \ref{t01}} and, moreover, possesses the symmetry fields
$\boldsymbol{u}_1$ and $\boldsymbol{u}_2$ which define the integrable
distribution and preserve the Hamiltonian
\begin{gather*}
[\boldsymbol{u}_1, \boldsymbol{v}]=[\boldsymbol{u}_2, \boldsymbol{v}]=0, \qquad \boldsymbol{u}_1(H)=\boldsymbol{u}_2(H)=0,\\
[\boldsymbol{u}_1, \boldsymbol{u}_2]=\lambda_1(\boldsymbol{x})\boldsymbol{u}_1 + \lambda_2(\boldsymbol{x})\boldsymbol{u}_2,
\end{gather*}
then the system~\eqref{eq01} admits the conformally Hamiltonian representation
\begin{gather*}
\dot{\boldsymbol{x}}=H^{(-1)}_u(\boldsymbol{x}){\bf J}^{(4)}(\boldsymbol{x})\frac{\partial H}{\partial \boldsymbol{x}}, \\
{\bf J}^{(4)}(\boldsymbol{x})=\boldsymbol{v}\wedge\boldsymbol{u} + \boldsymbol{u_1}\wedge\boldsymbol{u_2}, \qquad\operatorname{rank} {\bf J}^{(4)}=4.
\end{gather*}
\end{proposition}

Below we consider some nonholonomic mechanics problems illustrating the
applicability of the above theorem and allowing one to obtain Poisson
brackets with various (unusual) properties.

The following construction shows that any dynamical system admits a~rank~2 Hamiltonization
	in an extended phase space.

Let $\boldsymbol{v}$ denote the vector f\/ield that def\/ines the dynamics of a nonholonomic system on some manifold $\mathcal{M}$. Let $E\colon \mathcal{M} \rightarrow \mathbb{R}$ be the energy of the system, so that $\boldsymbol{v}(E) = 0$.
Consider the following rank~two Poisson structure in $\mathcal{M} \times \mathbb{R}$:
\begin{gather*}
{\bf J}=\boldsymbol{v}\wedge\frac{\partial}{\partial y},
\end{gather*}
where $y$ is the coordinate on the factor $\mathbb{R}$ of the Cartesian product~$\mathcal{M} \times \mathbb{R}$.
 The fact that~${\bf J}$ satisf\/ies the Jacobi identity is obvious from the
 fact that ${[\boldsymbol{v}, \frac{\partial}{\partial y}]=0}$.

 It is obvious that the Hamiltonian vector f\/ield of $y$ with respect to is $\boldsymbol{v}$.
 Moreover, $E$ is a~Casimir of~${\bf J}$.

This example shows how the symplectic leaves of a rank two Poisson structure can be immersed
in the phase space in a rather complicated way. Moreover, it also
illustrates that usually the construction of brackets in an extended phase
space cannot add anything new to understand the properties of a given
system.

\subsection{Invariant measure in Hamiltonian systems}

Let us brief\/ly discuss the problem of existence of an invariant measure
in Hamiltonian systems on $\mathcal{M}^n$ with a~Poisson structure
\begin{gather}
\label{mm}
\dot{\boldsymbol{x}}={\bf J}(\boldsymbol{x})\frac{\partial H}{\partial \boldsymbol{x}}, \qquad \operatorname{rank} {\bf J}<n.
\end{gather}
We recall that for a symplectic manifold the rank of ${\bf J}$ equals $n =
2k$ and that, according to the Liouville theorem, an invariant measure
exists for an arbitrary Hamiltonian. By analogy with this, if for the
system~\eqref{mm} with an arbitrary Hamiltonian there exists the same
invariant measure, it is called the {\it Liouville measure}. Its density
satisf\/ies the system of equations on~$\mathcal{M}^n$
\begin{gather*}
\sum_{j=1}^n\frac{\partial}{\partial x_j}(\rho J_{ij}(\boldsymbol{x}))=0, \qquad i=1,\ldots,n.
\end{gather*}
A manifold with this Poisson structure ${\bf J}(\boldsymbol{x})$ and volume form
$\rho(\boldsymbol{x})dx_1 \wedge \cdots \wedge dx_n$ is called a~{\it unimodular Poisson manifold}.

\begin{remark}
For the case of linear Poisson structures (Lie--Poisson brackets) the
unimodular Poisson manifolds correspond to the well-known unimodular Lie
algebras (see, e.g.,~\cite{KozIa}).
\end{remark}

Below we distinguish between three cases:
\begin{itemize}\itemsep=0pt
\item[--] {\it there is} an invariant measure (without singularities)
 which exists on the entire mani\-fold~$\mathcal{M}^n$ and is def\/ined
 by an explicit solution of the Liouville equation~\eqref{eq1_2};
\item[--] the Liouville equation~\eqref{eq1_2} admits an explicit
 solution for density $\rho(\boldsymbol{x})$ which has
 singularities at some points of the phase space; we shall call
 such a measure {\it singular};
\item[--] it is impossible to obtain an explicit Liouville solution,
 and in the phase space of the system there exist attracting
 invariant sets~-- attractors\footnote{As a rule, attractors can
 be detected only by numerical investigations.} (f\/ixed points,
 limit cycles, strange attractors etc.); in this case we shall say
 that {\it there exists no} invariant measure.
\end{itemize}

Further, we illustrate the above constructions and considerations by
various nonholonomic systems. Nonholonomic systems are characterized by
the presence of nonintegrable constraints, resulting in fairly general
systems of dif\/ferential equations, which in the case of homogeneous (in
velocities) constraints possess a f\/irst energy integral (see, e.g.,
\cite{Bizz} for details). The problem of existence of an invariant measure
for nonholonomic systems is discussed, for example, in~\cite{Bizz, b42, b43,
Kozlovv}.

In all the examples of Hamiltonizable systems considered here the
dimension of the phase space~$\mathcal{M}^n$ is equal to f\/ive ($n=5$).
Moreover, depending on the existence of an invariant measure and global
Casimir functions, a large number of types of Hamiltonian systems can
arise (with a~Poisson bracket of nonmaximal rank).
Nevertheless, as a rule, it is impossible to detect all these types in
practice. We list here combinations that are encountered in the examples
considered.
\begin{itemize}\itemsep=0pt
\item[1)] There exist an invariant measure and three global Casimir
 functions: a~nonholonomic hinge, one of the bodies is a plate (see
 Section~\ref{secpl});
\item[2)] There exist an invariant measure and two global Casimir
 functions: a~nonholonomic hinge in the general case
 (see Section~\ref{sh2});
\item[3)] There exist an invariant measure and one global Casimir
 function: the Suslov problem under the condition that the
 constraint is imposed along the principal axis of inertia (see
 Section~\ref{sysl01});
\item[4)] There exist a singular invariant measure and three global
 Casimir functions: the Chaplygin sleigh on a horizontal plane (see
 Section~\ref{sen01});
\item[5)] There are no invariant measure and no global Casimir
 functions: the Suslov problem in a gravitational f\/ield (see
 Section~\ref{sysl02}) and the Chaplygin sleigh on an inclined
 plane (see Section~\ref{sen02}).
\end{itemize}

\section{Nonholonomic hinge}\label{section2}
\subsection{Equations of motion}

Consider the problem of the free motion of a system of two bodies coupled
by a nonholonomic hinge. The outer body is a homogeneous spherical shell
inside which a rigid body moves; the rigid body is connected with the
shell by means of sharp wheels in such a way as to exclude relative
rotations about the vector~$\boldsymbol{e}$ f\/ixed in the inner body (see
Fig.~\ref{fig03}).

\begin{figure}[t]
\centering
\includegraphics{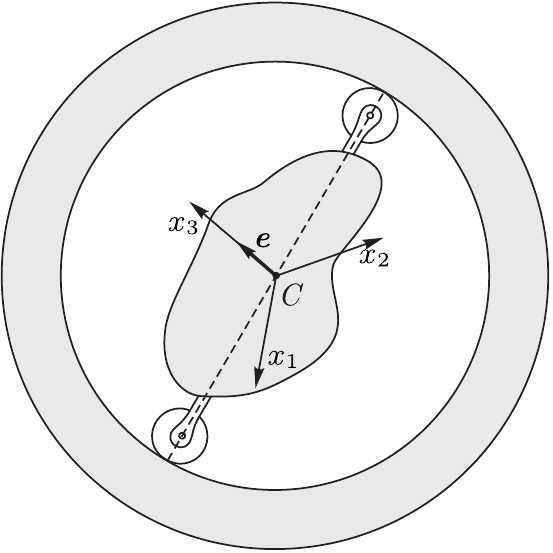}
\caption{A system of two bodies coupled by a nonholonomic hinge.}\label{fig03}
\end{figure}

In what follows we shall call this system a~{\it nonholonomic hinge}
(this problem was previously studied in~\cite{BBB, BBM02}).

Choose a moving coordinate system $Cx_1x_2x_3$ attached to the inner body.
Then, if we denote the angular velocities of the inner body and the spherical shell by $\boldsymbol{\omega}$ and $\boldsymbol{\Omega}$,
the constraint takes the form
\begin{gather*}
\omega_3-\Omega_3=0.
\end{gather*}
In the absence of external forces the evolution of the angular velocities is governed by the following equations:
\begin{gather}
\dot{\Omega}_1=\omega_3^{}(\Omega_2^{}-\omega_2^{}), \qquad \dot{\Omega}_2=\omega_3^{}(\omega_1^{}-\Omega_1^{}),\nonumber\\
I_1\dot{\omega}_1=(I_2 - I_3)\omega_2^{}\omega_3^{}, \qquad I_2\dot{\omega}_2=(I_3 - I_1)\omega_1^{}\omega_3^{} , \nonumber\\
(I_\text{s} + I_3)\dot{\omega}_3=I_\text{s}(\Omega_1\omega_2 - \Omega_2\omega_1) + (I_1 - I_2)\omega_1\omega_2,\label{Sy1}
\end{gather}
where ${ \bf I}=\operatorname{diag}(I_1^{},I_2^{},I_3^{})$ is
the tensor of inertia of the inner body and $I_\text{s}^{}$ is the tensor of inertia of the shell.

\subsection{First integrals and the Poisson bracket}\label{section1.2}
The system
\eqref{Sy1} possesses a standard invariant measure and conserves the energy:
\begin{gather*}
d\Omega_1^{} d\Omega_2^{} d\omega_1^{} d\omega_2^{} d\omega_3^{}, \qquad
E=\frac12I_\text{s}^{}(\Omega_1^2+\Omega_2^2)+\frac12\big(I_1^{}\omega_1^2+I_2^{}\omega_2^2+(I_3^{}+I_\text{s}^{})\omega_3^2\big).
\end{gather*}
In order to use Theorem \ref{t01}, we choose
\begin{gather*}
\boldsymbol{u}=\frac{1}{\omega_3}\frac{\partial}{\partial \omega_3}
\end{gather*}
as $\boldsymbol{u}$. Then, if we denote the initial vector f\/ield \eqref{Sy1} by $\boldsymbol{v}$, we obtain
\begin{gather*}
[\boldsymbol{v}, \boldsymbol{u}]=-\omega_3^{-1} \boldsymbol{v}, \qquad \boldsymbol{u}(E)=2(I_3^{}+I_\text{s}^{}).
\end{gather*}
Thus, Theorem \ref{t01} holds and the system \eqref{Sy1} can be
represented in Hamiltonian form with a~Poisson bracket in the variables
$\boldsymbol{x}=(\Omega_1, \Omega_2, \omega_1, \omega_2, \omega_3)$ of the
form
\begin{gather*}
{\bf J}^{(2)}=\boldsymbol{v}\wedge\boldsymbol{u}=\begin{pmatrix}
0& 0 & 0 & 0 & \Omega_2^{}-\omega_2^{} \\
0& 0 & 0 & 0 & -(\Omega_1^{}-\omega_1^{}) \vspace{1mm}\\
0& 0 & 0 & 0 & \frac{I_2^{}-I_3^{}}{I_1^{}}\omega_2^{}\vspace{1mm}\\
0& 0 & 0 & 0 & -\frac{I_1^{}-I_3^{}}{I_2^{}}\omega_1^{}\vspace{1mm}\\
*&*&*&*&0
\end{pmatrix} ,
\end{gather*}
where the asterisks denote nonzero matrix entries resulting from the
skew-symmetry condition~${\bf J}^{(2)}$. We note that even though the
vector f\/ield~$\boldsymbol{u}$ has a singularity when $\omega_3 = 0$, the
resulting Poisson structure for the system is smooth.

The found Poisson bracket corresponds to the {\it solvable Lie algebra}. According to the classif\/ication of~\cite{wint},
this is the algebra $A^{spq}_{5,17}$ with $p=q=0$ (see~\cite{BBM02} for details).

\subsection[A f\/lat inner body ($I_3=I_1+I_2$)]{A f\/lat inner body ($\boldsymbol{I_3=I_1+I_2}$)}\label{secpl}

In the case where the inner body is f\/lat, i.e., $I_3=I_1+I_2$, for
${\bf J}^{(2)}$ we f\/ind all
$k=n-\operatorname{rank} {\bf J}^{(2)}=3$ Casimir functions
\begin{gather*}
C_1(\boldsymbol{x})=\omega_1^2+\omega_2^2,\qquad
C_2(\boldsymbol{x})=\Omega_1^2+(\omega_2 - \Omega_2)^2,\qquad
C_3(\boldsymbol{x})=\Omega_2^2+(\omega_1 - \Omega_1)^2,
\end{gather*}
i.e., we obtain a case that is the simplest from the point of view of integrability.

Indeed, let us restrict the system \eqref{Sy1} to the symplectic leaf by
making the change of variables
\begin{gather*}
\omega_1=c_1\cos\varphi, \qquad \omega_2=c_1\sin\varphi, \\
\Omega_1=c_2\cos(\varphi + c_3), \qquad \Omega_2=c_1\sin\varphi -
c_2\sin(\varphi + c_3),
\end{gather*}
where $\varphi\in[0,2\pi)$ is the angular coordinate, and
 $c_1$, $c_2$, and $c_3$ parameterize the symplectic leaf
\begin{gather*}
C_1(\boldsymbol{x})=\sqrt{c_1}, \qquad C_2(\boldsymbol{x})=\sqrt{c_2}, \qquad C_3(\boldsymbol{x})=c_1^2 - 2c_1c_2\cos c_3 + c_2^2.
\end{gather*}
As a result, we obtain a system with one degree of freedom $(\varphi,
\omega_3)$ and a~canonical Poisson bracket
\begin{gather}
\label{Eqb03}
\dot{\varphi}=\frac{\partial \widetilde{H}}{\partial \omega_3}, \qquad \dot{\omega}_3=-\frac{\partial \widetilde{H}}{\partial \varphi},\nonumber\\
\widetilde{H}=\frac{1}{2}\omega_3^2 + \frac{c_1^2(I_1\cos^2\varphi + (I_s + I_2)\sin^2\varphi)}{2(I_s + I_1 + I_2)} - \frac{I_s c_1c_2\sin(\varphi + c_3)}{I_s + I_s + I_2}.
\end{gather}
Some of its typical trajectories are shown in Fig.~\ref{fig003}.

\begin{figure}[t]
\centering
\includegraphics[totalheight=4cm]{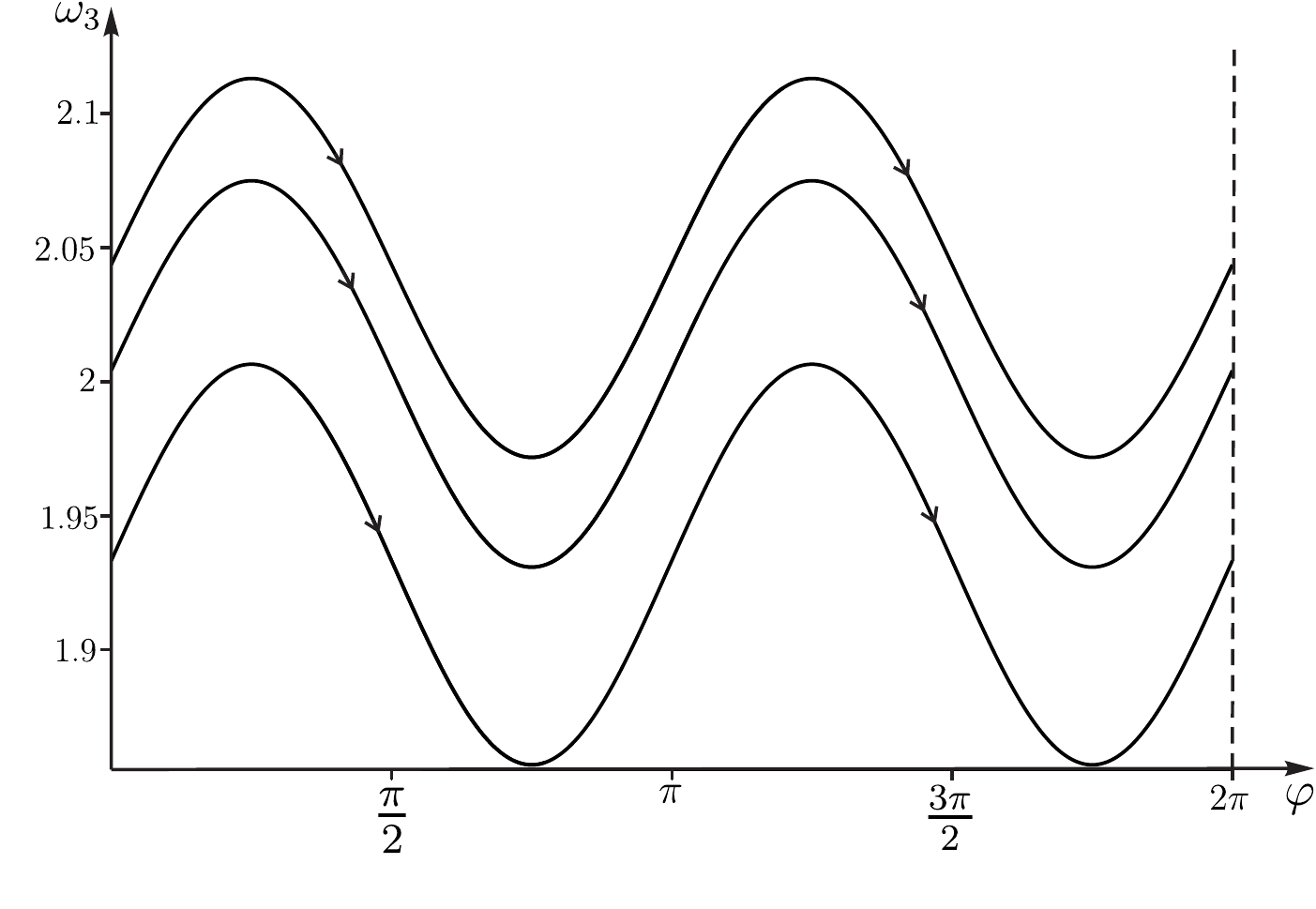}
\caption{Typical trajectories $\eqref{Eqb03}$ for $I_1=2$, $I_2=3$, $I_{\text{s}}=1$, $c_1=1$,
$c_2=2$, $c_3=\frac{\pi}{3}$.}\label{fig003}
\end{figure}

\subsection{The general case}\label{sh2}

In the general case the Poisson bracket possesses only two global
Casimir functions~\cite{BBB, BBM02}
\begin{gather*}
C_1(\boldsymbol{x})=I_1(I_1-I_3)\omega_1^2 + I_2(I_2-I_3)\omega_2^2,
\qquad C_2(\boldsymbol{x})=(I_1\omega_1 - I_3\Omega_1)^2 + (I_2\omega_2 - I_3\Omega_2)^2,
\end{gather*}
while a third one does not exist in the general case (it is def\/ined only
locally).

Indeed, let us introduce coordinates on the symplectic leaf and assume that $I_3<I_2<I_1$:
\begin{gather*}
\omega_1=\sqrt{\frac{c_1}{I_1(I_1 - I_3)}}\sin\varphi_1, \qquad \omega_2=\sqrt{\frac{c_1}{I_2(I_2 - I_3)}}\cos\varphi_1, \\
\Omega_1=\sqrt{\frac{I_1c_1}{I_1 - I_3}}\frac{\sin\varphi_1}{I_3} - \sqrt{c_2}\frac{\sin\varphi_2}{I_3}, \qquad
\Omega_2=\sqrt{\frac{I_2c_1}{I_2 - I_3}}\frac{\cos\varphi_1}{I_3} - \sqrt{c_2}\frac{\cos\varphi_2}{I_3},
\end{gather*}
where $\varphi_1,\varphi_2\in[0,2\pi)$ and $C_1(\boldsymbol{x})=c_1$, $C_2(\boldsymbol{x})=c_2$.

In addition, we restrict the system to the level set of the energy integral $H=h$ by making the substitution
\begin{gather*}
\omega_3(\varphi_1,\varphi_2)=\pm \sqrt{2h - \frac{2Q}{I_3^2(I_3 + I_\text{s})}}, \\
Q=\frac{I_\text{s}c_2}{2} + \frac{c_1}{2}\left( \frac{I_\text{s}I_1 + I_3^2}{I_1 - I_3}\sin^2\varphi_1 +
\frac{I_\text{s}I_2 + I_3^2}{I_2 - I_3}\cos^2\varphi_1 \right) \\
\hphantom{Q=}{}
- I_\text{s}\sqrt{c_2}\left(\sqrt{\frac{I_1c_1}{I_1 - I_3}} \sin\varphi_1\sin\varphi_2 +
\sqrt{\frac{I_2c_1}{I_2 - I_3}}\cos\varphi_1\cos\varphi_2 \right).
\end{gather*}
As a result, the equations of motion can be represented as
\begin{gather}
\label{eqB07}
\dot{\varphi}_1=k\omega_3(\varphi_1,\varphi_2), \qquad \dot{\varphi}_2=\omega_3(\varphi_1,\varphi_2), \qquad k^2=\frac{(I_1 - I_3)(I_2 - I_3)}{I_1I_2}.
\end{gather}

In the case where $\omega_3(\varphi_1,\varphi_2)$ vanishes nowhere and $k$
is irrational, the trajectories~\eqref{eqB07} are rectilinear orbits on a
torus (see Fig.~\ref{fig05}).

\begin{figure}[t]
\centering
\includegraphics[totalheight=5cm]{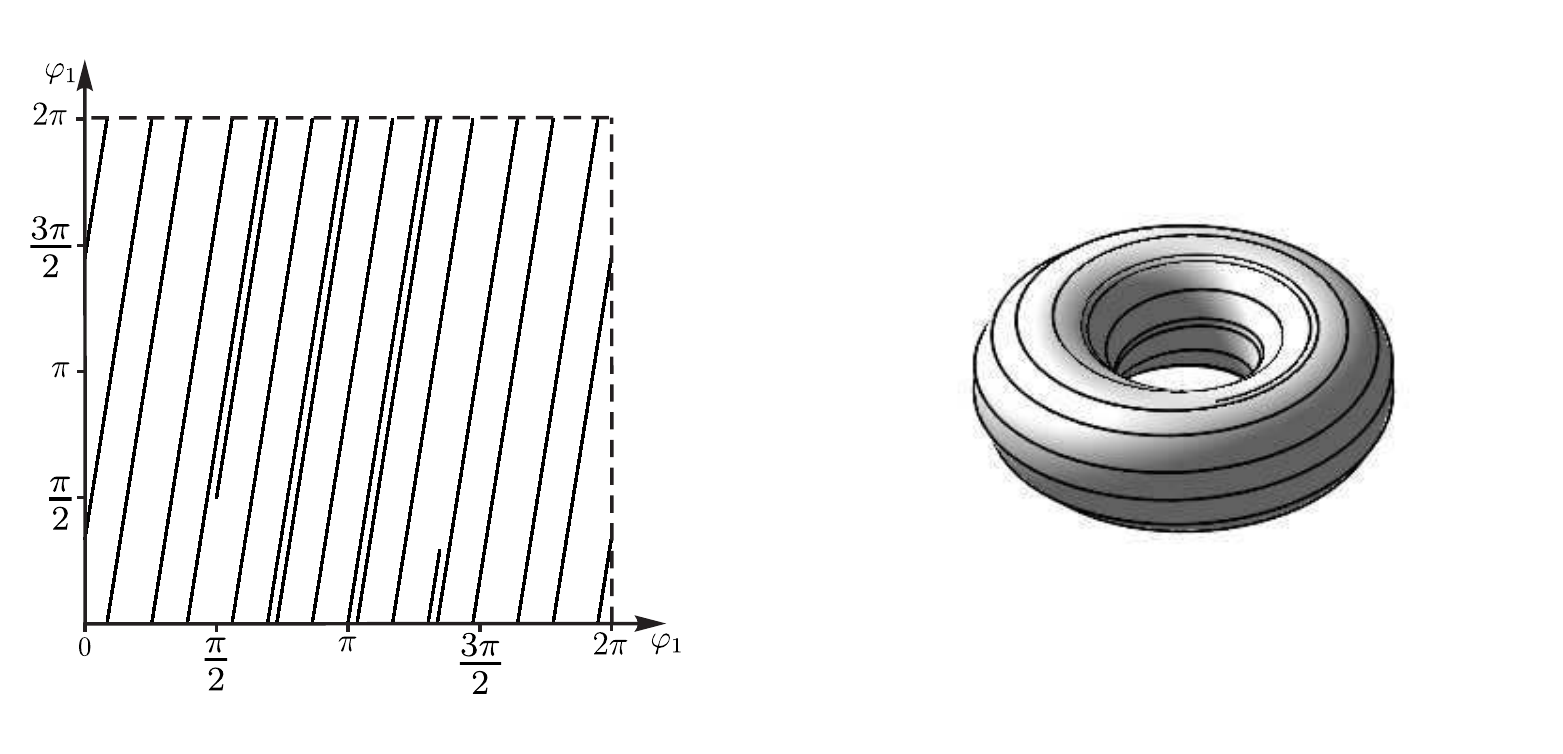}
\caption{Trajectory \eqref{eqB07} for $I_1=5.2$, $I_2=4.3$, $I_3=3.9$, $I_{\text{s}}=3.6$, $c_1=10$,
$c_2=40$, $h=10$.}\label{fig05}
\end{figure}

Such behavior is partially due to the fact that in the general case the
missing solution~\eqref{ma01} in the chosen (local) coordinates is
represented as
\begin{gather*}
\varphi_1 - k \varphi_2={\rm const},
\end{gather*}
whence it follows that the projection of the symplectic leaf on
($\varphi_1, \varphi_2$) is multi-valued.

A topological analysis of the system \eqref{Sy1} is presented in~\cite{BBB}. In particular, it is shown that in addition to a torus there
are two other types of integral surfaces: a~sphere and a sphere with three
handles (two-dimensional orientable surface of genus~3).

\section{The Suslov problem}\label{section3}
\subsection{Equations of motion and f\/irst integrals}
\label{sec02}

If we assume the spherical shell from the preceding example to be f\/ixed, we obtain
the so-called Suslov problem (in Vagner's interpretation, see Fig.~\ref{figss}),
which describes the motion of a rigid body with a f\/ixed point subject to the nonholonomic constraint
\begin{gather*}
(\boldsymbol{\omega}, \boldsymbol{e})=0,
\end{gather*}
where $\boldsymbol{\omega}$ is the angular velocity of the body
and $\boldsymbol{e}$ is the vector f\/ixed in the body.

\begin{figure}[t]
\centering
\includegraphics{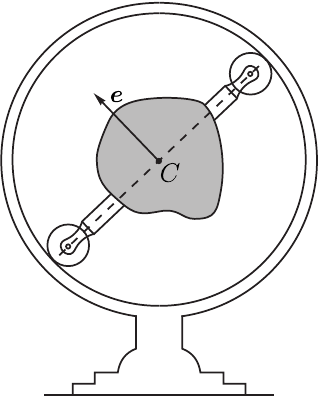}
\caption{Realization of the Suslov problem.}
\label{figss}
\end{figure}

Choose a moving coordinate system attached to the inner body, with one of the axes directed along $\boldsymbol{e}$.
Then the tensor of inertia of the moving body can be represented as
\begin{gather*}
{\bf I}=\begin{pmatrix}
I_{11}& 0 & I_{13} \\
0& I_{22} & I_{23} \\
I_{13}& I_{23} & I_{33} \\
\end{pmatrix} .
\end{gather*}
In the axisymmetric potential f\/ield $U=U(\boldsymbol{\gamma})$ the problem reduces to the investigation of the following closed system of equations~\cite{BKM}:
\begin{gather}
I_{11}\dot\omega_1=-\omega_2(I_{13}\omega_1+I_{23}\omega_2)+
\gamma_2\frac{\partial U}{\partial\gamma_3}-
\gamma_3\frac{\partial U}{\partial\gamma_2},\nonumber\\
I_{22}\dot\omega_2=\omega_1(I_{13}\omega_1+I_{23}\omega_2)+
\gamma_3\frac{\partial U}{\partial\gamma_1}-
\gamma_1\frac{\partial U}{\partial\gamma_3},\nonumber\\
\dot\gamma_1=-\gamma_3\omega_2,\qquad
\dot\gamma_2=\gamma_3\omega_1,\qquad
\dot\gamma_3=\gamma_1\omega_2-\gamma_2\omega_1,\label{ss01}
\end{gather}
where $\boldsymbol{\gamma}=(\gamma_1, \gamma_2,\gamma_3)$ is the unit vector of the f\/ixed coordinate system referred to the moving axes.

The system~\eqref{ss01} conserves the energy and possesses the geometrical integral
\begin{gather*}
E=\frac{1}{2}\big(I_{11} \omega_1^2 + I_{22} \omega_2^2\big) + U, \qquad F=\gamma_1^2 + \gamma_2^2 + \gamma_3^2=1.
\end{gather*}

We now consider successively several particular cases~\eqref{ss01}.

\subsection[Case $I_{13}=I_{23}=0$ and $U=U(\gamma_1, \gamma_2)$]{Case $\boldsymbol{I_{13}=I_{23}=0}$ and $\boldsymbol{U=U(\gamma_1, \gamma_2)}$}\label{sysl01}

Suppose that $I_{13}=I_{23}=0$, i.e., {\it the constraint is imposed along the principal axis of inertia} and the potential f\/ield has the form $U=U(\gamma_1, \gamma_2)$. Then
the system~\eqref{ss01} possesses a standard invariant measure.

Further, we introduce the vector f\/ield
\begin{gather*}
\boldsymbol{u}=\frac{1}{\gamma_3}\frac{\partial}{\partial \gamma_3},
\end{gather*}
for which
\begin{gather*}
[\boldsymbol{v}, \boldsymbol{u}]=-\gamma_3^{-1}\boldsymbol{v}, \qquad \boldsymbol{u}(F)=2.
\end{gather*}
Consequently, it follows from Theorem~\ref{t01} that the initial system~\eqref{ss01} can be represented in Hamiltonian form with the Poisson bracket
\begin{gather*}
{\bf J}^{(2)}=\boldsymbol{v}\wedge\boldsymbol{u}=\begin{pmatrix}
0& 0 & 0 & 0 & -\frac{1}{I_{11}}\frac{\partial U}{\partial \gamma_2} \vspace{1mm}\\
0& 0 & 0 & 0 & \frac{1}{I_{22}}\frac{\partial U}{\partial \gamma_1} \\
0& 0 & 0 & 0 & -\omega_2\\
0& 0 & 0 & 0 & \omega_1\\
*&*&*&*&0
\end{pmatrix} ,
\end{gather*}
where the asterisks denote nonzero matrix entries resulting from the
skew-symmetry condition~${\bf J}^{(2)}$.

\begin{proposition}
In the general case, the bracket ${\bf J}^{(2)}$ possesses one global Casimir function~-- the energy integral~$E$.
\end{proposition}

\begin{proof}
First of all, we note that the Casimir functions ${\bf J}^{(2)}$ are simultaneously the integrals of the system
\begin{gather}\label{ss0111}
\dot{\omega}_1=-\frac{1}{I_{11}}\frac{\partial U}{\partial \gamma_2}, \qquad
\dot{\omega}_2= \frac{1}{I_{22}}\frac{\partial U}{\partial \gamma_1}, \qquad
\dot{\gamma}_1=-\omega_2, \qquad \dot{\gamma}_2=\omega_1,
\end{gather}
which is Hamiltonian with Hamiltonian $E$ and the Poisson bracket
\begin{gather*}
\{\omega_1,\gamma_2\}=-\frac{1}{I_{1}}, \qquad \{\omega_2,\gamma_1\}=\frac{1}{I_{2}}.
\end{gather*}
Further, if we introduce new variables
\begin{gather*}
p_1=\sqrt{I_1}\omega_1, \qquad p_2=\sqrt{I_2}\omega_2, \qquad
q_1=\sqrt{I_2}\gamma_1, \qquad q_2=\sqrt{I_1}\gamma_2,
\end{gather*}
then the system \eqref{ss0111} reduces to investigating a vector f\/ield
with a canonical Poisson bracket and a Hamiltonian of the form
\begin{gather}\label{EEqq1}
H=\frac{1}{2}(p_1^2 + p_2^2) + U\left(\frac{q_1}{I_2}, \frac{q_2}{I_1}\right) .
\end{gather}
This is a natural system which describes the motion of a material point on
a~plane and in which, as is well known, there are generally no additional
f\/irst integrals, and hence the resulting Poisson bracket has no remaining
global Casimir functions.
\end{proof}

\begin{remark}
However, the system \eqref{EEqq1} has the well-know potentials $U$ for
which there exist (one or two) additional integrals (and hence Casimir
functions) of dif\/ferent degrees in momenta (see, e.g., the recent papers
\cite{YosM, YosM2} or \cite{Henon02, Henon01} for the
H\'{e}non--Heiles system).
\end{remark}

\subsection[Case $I_{13}\neq0$, $I_{23}\neq0$ and $U=(a, \gamma)$]{Case $\boldsymbol{I_{13}\neq0}$, $\boldsymbol{I_{23}\neq0}$ and $\boldsymbol{U=(\boldsymbol{a}, \boldsymbol{\gamma})}$ }
\label{sysl02}

Now consider the more general situation $I_{13}\neq0$, $I_{23}\neq0$, but in the potential f\/ield $U=(\boldsymbol{a}, \boldsymbol{\gamma})$.

In this case, the equations of motion are invariant under the transformation
\begin{gather*}
\omega_1 \rightarrow \lambda \omega_1, \qquad \omega_2 \rightarrow \lambda \omega_2, \qquad
\boldsymbol{\gamma} \rightarrow \lambda^2 \boldsymbol{\gamma}, \qquad dt \rightarrow \lambda^{-1} dt,
\end{gather*}
which is associated to the vector f\/ield
\begin{gather*}
\widehat{\boldsymbol{u}}=\omega_1\frac{\partial}{\partial \omega_1 } +
\omega_2\frac{\partial}{\partial \omega_2 } +
2\gamma_1\frac{\partial}{\partial \omega_1 } + 2\gamma_2\frac{\partial}{\partial \omega_2 } +
2\gamma_3\frac{\partial}{\partial \omega_3 },
\end{gather*}
for which we have
\begin{gather*}
[\boldsymbol{u}, \boldsymbol{v}]=\boldsymbol{v}, \qquad \boldsymbol{u}(E)=2E.
\end{gather*}
Thus, the system under consideration can be represented in conformally Hamiltonian form
with the Poisson bracket
\begin{gather*}
{\bf J}^{(2)}=\boldsymbol{v}\wedge\boldsymbol{u}.
\end{gather*}
In the case at hand, ${\bf J}^{(2)}$ has a rather cumbersome form, so we do not present it here explicitly.
We note that it is cubic in the velocities $\boldsymbol{\omega}$.

\begin{proposition}
In the general case, the system~\eqref{ss01}
has no invariant measure with smooth density in the potential field $U=(\boldsymbol{a}, \boldsymbol{\gamma})$.
\end{proposition}

\begin{proof} In the potential f\/ield $U=(\boldsymbol{a}, \boldsymbol{\gamma})$ the system~\eqref{ss01} possesses a f\/ixed point
\begin{gather*}
\omega_1=0, \qquad \omega_2=-\frac{a_3}{I_{23}}, \qquad \gamma_1=0, \qquad \gamma_2=1, \qquad \gamma_3=0
\end{gather*}
having a nonzero trace of the linearization matrix:
\begin{gather*}
\frac{I_{13}}{I_1}\frac{\sqrt{a_3}}{\sqrt{I_{23}}}.
\end{gather*}
Consequently, the system considered has no invariant measure with smooth density \cite{Kozlovv, Kozll, Kozlovv2} in a neighborhood of the f\/ixed point.

Indeed, in this case the characteristic polynomial of the linearized system is represented as
\begin{gather*}
P(\lambda)=\lambda^2\big( \lambda^3 - c_1\lambda^2 + c_2\lambda + c_3 \big), \\
 c_1=\frac{I_{13}\sqrt{a_3}}{I_1\sqrt{I_{23}}}, \qquad c_2=\frac{a_2}{I_1} -
\frac{I_1I_2+2I_{23}^2}{I_{23}I_1I_2}a_3, \qquad
c_3=\sqrt{a_3I_{23}}\left(\frac{I_{13}a_3}{I_1I_{23}^2} -
\frac{2a_1}{I_1I_2}\right) .
\end{gather*}
This polynomial is not reciprocal.
\end{proof}

In this system, an analogous result was obtained for $U=0$ previously
in~\cite{Kozlovv}. Later, however, a singular invariant measure \cite{BKM}
was found in this case.

We also note that the pair of brackets that was constructed for the Suslov problem is def\/ined in the extended phase space, as are the brackets described above.
\begin{gather*}
\mathbb{R}^2\times \mathbb{R}^3=\{(\omega_1, \omega_2, \gamma_1, \gamma_2, \gamma_3 )\}.
\end{gather*}
Indeed, for these brackets the geometric integral $F=\gamma_1^2 +\gamma_2^2 + \gamma_3^2$
is not a Casimir function, so it is not clear how to restrict them to the actual phase space of the system
\begin{gather*}
\mathbb{R}^2\times S^2=\{(\omega_1, \omega_2, \gamma_1, \gamma_2, \gamma_3 ), \gamma_1^2 +\gamma_2^2 + \gamma_3^2=1\}.
\end{gather*}
In view of the above
discussion (at the end of Section~\ref{section1.2}), these brackets cannot provide any insight into the system's properties.

\begin{remark}
It may seem that a singular invariant measure with density having
singularities in the considered family of f\/ixed points is possible in this
case too (see, e.g., Section~\ref{section4}). However, it is impossible to obtain a
suitable solution of the Liouville equation in this case. Moreover, as
numerical investigations show, this system may also have more complicated
attracting sets~-- limit cycles. Numerical investigations also show that
this system has no additional real-analytical integrals.
\end{remark}

\begin{remark}
A proof of the absence of meromorphic integrals for other (more
particular) cases of the Suslov problem was obtained in \cite{Fed,
Mahdi, Zig}. The methods developed in these papers can be used in this case
also, since the system \eqref{ss01} has a particular periodic solution
which for $I_{13}=0$, $a_1=0$, and $a_2=0$ has the form
\begin{gather*}
\omega_1(t)=2k \operatorname{cn}(t,k),\quad \omega_2(t)=0, \qquad \gamma_1(t)=0, \qquad \gamma_2(t)=-\frac{2I_1k}{a_3}\operatorname{sn}(t,k)\operatorname{dn}(t,k),\\
\gamma_3(t)=\frac{I_1}{a_3}\big(2k^2\operatorname{sn}^2(t,k) - 1\big), \qquad H=I_1\big(2k^2-1\big),
\end{gather*}
where $\operatorname{sn}(t,k)$, $\operatorname{cn}(t,k)$, $\operatorname{dn}(t,k)$ are the elliptic
Jacobi functions with parameter~$k$.
\end{remark}

\begin{remark}
We note that in \cite{BKM, Kozlovv2} other cases of representation of the
equations of the Suslov problem were found in Hamiltonian form, with a~Poisson bracket of rank~4, which cannot be obtained with the help of the
Hojman construction.
\end{remark}

\section{The Chaplygin sleigh}\label{section4}
\subsection{Equations of motion}
\label{sec03}

The Chaplygin sleigh is a rigid body moving on a horizontal plane and
supported at three points: two (absolutely) smooth posts and a knife edge
such that the body cannot move perpendicularly to the plane of the wheel.

As a historical remark, we note that although the Chaplygin sleigh is
usually linked to the works of Chaplygin \cite{b21-3} and Caratheodory~\cite{Caratheodory}, they were considered slightly earlier by Brill~\cite{Brill} as an example of the mechanism of a nonholonomic planimeter.

\begin{figure}[t]
\centering
\includegraphics[totalheight=5cm]{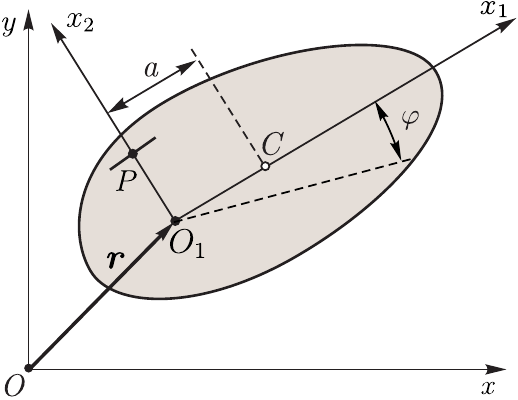}
\caption{The Chaplygin sleigh on a plane.}
\label{fig01}
\end{figure}

In this case, the conf\/iguration space coincides with the group of motions
of the plane ${\rm SE}(2)$. To parameterize it, we choose Cartesian coordinates
$(x,y)$ of point $O_1$ (see Fig.~\ref{fig01}) in the coordinate system
$Oxy$ and the angle $\varphi$ of rotation of the axes $O_1x_1x_2$ relative
to $Oxy$.

The motion of the Chaplygin sleigh in the variables $(v_1,\omega,
x, y,\varphi)$ is described by the equations (see \cite{BM} for details)
\begin{gather}
m\dot v_1=ma\omega^2-\frac{\partial U}{\partial x}\cos\varphi-\frac{\partial
U}{\partial y}\sin\varphi,\qquad(I+ma^2)\dot\omega=-ma\omega v_1-\frac{\partial
U}{\partial\varphi},\nonumber\\
\dot x=v_1\cos\varphi,\qquad\dot y=v_1\sin\varphi, \qquad\dot\varphi=\omega,\label{eq27052-3}
\end{gather}
where $v_1$ is the projection of the velocity of point $O_1$ on the axis
$O_1x_1$, and~$m$ and~$I$ are, respectively, the mass and the moment of
inertia of the body, $a$~is the distance specifying the position of the
knife edge, and~$U$ is the potential of the external forces.

The system \eqref{eq27052-3} conserves the energy integral
\begin{gather*}
E=\frac12\big(mv_1^2+\big(I+ma^2\big)\omega^2\big) + U.
\end{gather*}

Below we consider successively several particular cases.

\subsection{The Chaplygin sleigh on a horizontal plane}
\label{sen01}

{\bf Symmetry f\/ields and invariant measure.} In the absence of an external
potential f\/ield $U=0$ the system~\eqref{eq27052-3} admits the symmetry
group $SE(2)$ which is associated to the vector f\/ields
\begin{gather*}
\boldsymbol{u}_x=\frac{\partial}{\partial x}, \qquad \boldsymbol{u}_y=\frac{\partial}{\partial y}, \qquad
\boldsymbol{u}_{\varphi}=\frac{\partial}{\partial \varphi},
\end{gather*}
which are symmetry f\/ields.

Moreover, \eqref{eq27052-3} possess a singular invariant measure
\begin{gather*}
\omega^{-1}dv_1d\omega dx dy d\varphi.
\end{gather*}

{\bf Poisson bracket.}
We note that the system \eqref{eq27052-3} is invariant under the transformation
\begin{gather*}
v_1 \rightarrow \lambda v_1, \qquad \omega \rightarrow \lambda \omega, \qquad
dt \rightarrow \lambda^{-1} dt,
\end{gather*}
which is associated to the vector f\/ield
\begin{gather*}
\boldsymbol{u}=v_1\frac{\partial}{\partial v_1} + \omega\frac{\partial}{\partial \omega},
\end{gather*}
and
\begin{gather*}
[\boldsymbol{u},\boldsymbol{v}]=\boldsymbol{v}, \qquad \boldsymbol{u}(E)=2E,
\end{gather*}
where $\boldsymbol{v}$ is the initial vector f\/ield def\/ined by the system~\eqref{eq27052-3}.

Consequently, in the case $E\neq0$ the system \eqref{eq27052-3} is
represented in Hamiltonian form
\begin{gather*}
\dot{\boldsymbol{x}}={\bf J}^{(2)}\frac{\partial E}{\partial \boldsymbol{x}}, \qquad
\boldsymbol{x}=(v_1, \omega, x, y, \varphi),
\end{gather*}
with the Poisson bracket
\begin{gather*}
{\bf J}^{(2)}=\frac{\boldsymbol{v}\wedge\boldsymbol{u}}{2E}=\begin{pmatrix}
0& \frac{a\omega}{I + ma^2} & - \frac{v_1^2\cos \varphi}{mv_1^2 + (I + ma^2)\omega^2} & - \frac{v_1^2\sin \varphi}{mv_1^2 + (I + ma^2)\omega^2} &
-\frac{\omega v_1}{mv_1^2 + (I + ma^2)\omega^2}\\
*& 0 & -\frac{v_1\omega\cos\varphi}{mv_1^2 + (I + ma^2)\omega^2} & -\frac{v_1\omega \sin \varphi}{mv_1^2 + (I + ma^2)\omega^2} & -\frac{\omega^2}{mv_1^2 + (I + ma^2)\omega^2} \\
*& * & 0 & 0 & 0\\
*& * & * & 0 & 0\\
*&*&*&*&0
\end{pmatrix} ,
\end{gather*}
where the asterisks denote the matrix entries resulting from the skew-symmetry condition
${\bf J}^{(2)}$.

We examine in more detail a symplectic foliation which is def\/ined by the bracket ${\bf J}^{(2)}$.

First of all, we note that the manifold
\begin{gather*}
\mathcal{M}_0^4=\{ (\omega, v_1, x, y, \varphi) \,| \, \omega=0 \}
\end{gather*}
def\/ines a Poisson submanifold with the Casimir functions
\begin{gather*}
C_1=\varphi, \qquad C_2=x\sin\varphi - y\cos \varphi.
\end{gather*}
Thus, the entire phase space of the system $\mathcal{M}^5$ is a union of three nonintersecting Poisson submanifolds
\begin{gather*}
\mathcal{M}^5=\mathcal{M}^5_+\cup\mathcal{M}_0^4\cup\mathcal{M}^5_-, \\
\mathcal{M}^5_+=\{(\omega, v_1, x, y, \varphi)\,| \, \omega >0\}, \qquad
\mathcal{M}^5_-=\{(\omega, v_1, x, y, \varphi)\,| \, \omega <0\}.
\end{gather*}
Let us consider $\mathcal{M}^5_+$ in more detail and pass from $(v_1, \omega)$
to (polar) coordinates $(h, \psi)$:
\begin{gather*}
v_1=Aa h\cos\psi, \qquad \omega=h\sin\psi, \qquad
A^2=1+\frac{I}{ma^2}>1, \qquad E=\frac{1}{2}mA^2a^2h^2,
\end{gather*}
where $h\in(0, \infty)$ and $\psi\in(0,\pi)$. In the new variables $(\psi, x,y,\varphi, h)$ the Poisson bracket
becomes
\begin{gather}
\label{eqB1}
{\bf J}^{(2)}=\begin{pmatrix}
0& 0 & 0 & 0 & -hA^{-1}\sin\psi\\
0& 0 & 0 & 0 & hAa\cos\psi\cos\varphi \\
0& 0 & 0 & 0 & hAa\cos\psi\sin\varphi\\
0& 0 & 0 & 0 & h\sin\psi\\
*&*&*&*&0
\end{pmatrix} .
\end{gather}
As a result, the Casimir functions \eqref{eqB1} have the form
\begin{gather}
C_1=A\psi + \varphi, \qquad
C_2=x + A^2a\int_{\psi^*}^\psi \tan(\xi)\cos(A\xi - C_1)d\xi, \nonumber\\
C_3=y - A^2a\int_{\psi^*}^\psi \tan(\xi)\sin(A\xi - C_1)d\xi,\label{eqB34}
\end{gather}
where $\psi^*\in(0,\pi)$.

The level surface of the Casimir function $C_1={\rm const}$ foliates the
domain $\mathcal{M}^5_+$ into four-dimensional surfaces. The projection of
each such surface into the three-dimensional space
\begin{gather*}
S^1\times \mathbb{R}^2_+=\{ (\varphi, \omega, v_1)\,| \, \omega>0 \}
\end{gather*}
is a ruled surface (similar to a helicoid). Its equation can be
represented in the following parametric form
\begin{gather}
\label{eqB33}
v_1=Aah\cos\left(\frac{\varphi + C_1 }{A}\right) , \qquad
\omega=h\sin\left(\frac{\varphi + C_1 }{A}\right) , \qquad \varphi=\varphi.
\end{gather}
 In the def\/inition~\eqref{eqB33} the periodicity in $\varphi$ should, of
 course, be taken into account (see Fig.~\ref{fig0001}).

\begin{figure}[t]
\centering
\includegraphics[totalheight=5.5cm]{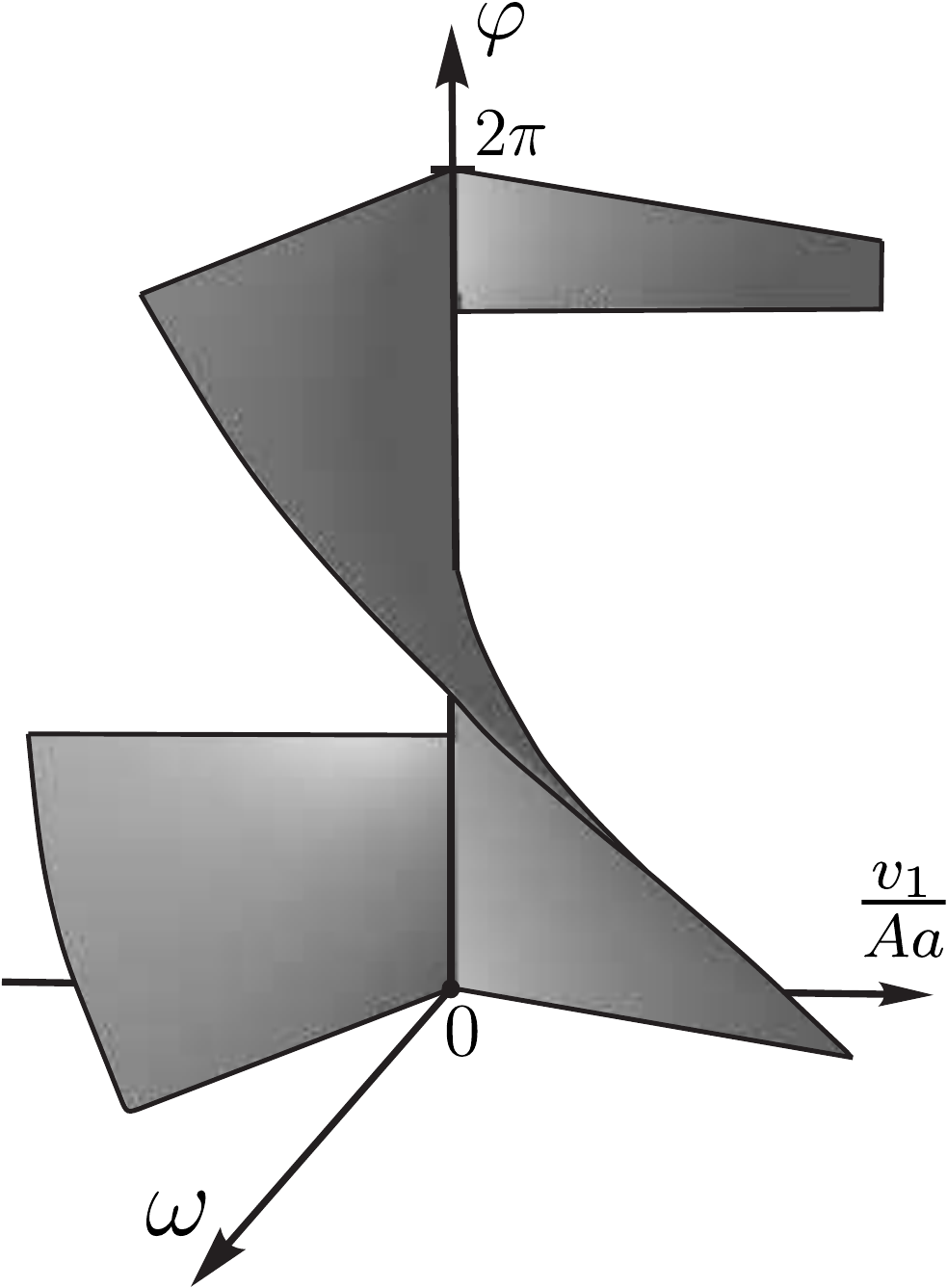}
\caption{Characteristic surface def\/ined by the Casimir function $C_1$ in
the domain $\mathcal{M}^5_+$ for $C_1=2$ and $A=3$.} \label{fig0001}
\end{figure}

Proceeding in a similar way for $\mathcal{M}^5_-$, it is easy to check
that relations~\eqref{eqB34} also def\/ine the Casimir functions with the
only dif\/ference that $\psi^*\in(\pi,2\pi)$.

{\bf The second Poisson bracket.} In this case it turns out that when
$E\neq0$, by Proposi\-tion~\ref{pro1}, the system~\eqref{eq27052-3} can be
represented in conformally Hamiltonian form
\begin{gather*}
\dot{\boldsymbol{x}}=(2E)^{(-1)}{\bf J}^{(4)}\frac{\partial E}{\partial \boldsymbol{x}},
\end{gather*}
with a Poisson bracket of rank~4 using the
symmetry f\/ields as follows: $\widehat{\boldsymbol{u}}_x$ and
$\widehat{\boldsymbol{u}}_y$
\begin{gather*}
{\bf J}^{(4)} = \boldsymbol{v}\wedge\boldsymbol{u} + \boldsymbol{u}_x\wedge\boldsymbol{u}_y =
\!\!\begin{pmatrix}
0& a\omega^3 + \frac{ma}{I + ma^2}\omega v_1^2 & - v_1^2\cos \varphi & - v_1^2\sin \varphi & -\omega v_1\\
*& 0 & -v_1\omega\cos\varphi & -v_1\omega \sin \varphi & -\omega^2 \\
*& * & 0 & 1 & 0\\
*& * & * & 0 & 0\\
*&*&*&*&0
\end{pmatrix} .
\end{gather*}

\begin{remark}
We note that the Jacobi identity for the matrix $\frac{{\bf J}^{(4)}}{E}$
does not hold.
\end{remark}

In this case, $\omega=0$ def\/ines the Poisson submanifold ${\bf J}^{(4)}$
on which ${\operatorname{rank} {\bf J}^{(4)}=2}$, and the Casimir functions are
\begin{gather*}
C_1=\varphi, \qquad C_2=Aa(x\sin\varphi - y \cos\varphi) - h^{-1}.
\end{gather*}
In the remaining region of the phase space $\overline{{\bf J}}$ possesses
the only Casimir function
\begin{gather*}
C_1=A\psi + \varphi, \qquad \tan\varphi=\frac{v_2}{Aa\omega}.
\end{gather*}

The bracket ${\bf J}^{(4)}$ was obtained previously in the initial
coordinates $(v_1,\omega, x, y, \varphi)$ in \cite{BM}.

\subsection{The Chaplygin sleigh on an inclined plane}
\label{sen02}

Below we consider the motion of the Chaplygin sleigh on an inclined plane.
We direct the axis~$Ox$ along the line of maximum slope, then
\begin{gather*}
U=m\mu (x +a\cos \varphi), \qquad \mu=g \sin \chi,
\end{gather*}
where $\chi$ is the angle of inclination of the plane to the horizon.
In this case, the energy integral has the form
\begin{gather*}
E=\frac{I + ma^2}{2}\omega^2 + \frac{mv_1^2}{2} + m\mu (x +a\cos \varphi).
\end{gather*}
Moreover, two symmetry f\/ields are preserved in the system~\eqref{eq27052-3}
\begin{gather*}
\boldsymbol{u}_x=\frac{\partial}{\partial x}, \qquad \boldsymbol{u}_y=\frac{\partial}{\partial y}.
\end{gather*}
Since $\boldsymbol{u}_x (E)=m\mu$, Theorem 1 applies to this
system.

As a result, its equations of motion are represented in Hamiltonian form
with the Poisson bracket ${\bf J}^{(2)}$ of rank 2:
\begin{gather*}
{\bf J}^{(2)}=\boldsymbol{v}\wedge\boldsymbol{u}_x=\begin{pmatrix}
0& \omega & 0 & 0 & 0\\
*& 0 & -v_1\sin\varphi & \mu \cos\varphi - a\omega^2 & \frac{ma}{I + ma^2}(\omega v_1-\mu\sin\varphi) \\
*& * & 0 & 0 & 0\\
*& * & * & 0 & 0\\
*&*&*&*&0
\end{pmatrix} ,
\end{gather*}
where $\boldsymbol{v}$ is the initial vector f\/ield def\/ined by the system~\eqref{eq27052-3}.

In \cite{BM} it is shown that the dynamics of this system, which is
noncompact, has a complicated, seemingly random behavior, leading to the
absence of an invariant measure and global analytic Casimir functions in
this system.

\section{Conclusion}\label{section5}

Thus, in this paper it has been shown that the Hojman construction (using
conformal symmetry f\/ields) allows one to obtain Poisson brackets for
various nonholonomic systems. As a rule, the resulting Poisson brackets do
not possess a maximal rank, and in the general case a smooth invariant
measure and global Casimir functions may be absent for these brackets.
Nevertheless, as shown in~\cite{BBB}, the above-mentioned Poisson brackets
can be useful to investigate stability problems.

\subsection*{Acknowledgments}
Section~\ref{section2} was prepared by A.V.~Borisov under the RSF
grant No.~15-12-20035. Section~\ref{section1} was written by I.S.~Mamaev within the
framework of the state assignment for institutions of higher education.
The work of I.A.~Bizyaev (Sections~\ref{section3} and~\ref{section4}) was supported by RFBR grant
No.~15-31-50172.
The authors thank A.V.~Tsiganov and A.V.~Bolsinov for useful
discussions and the referees for numerous comments, which have contributed
to the improvement of this paper.

\pdfbookmark[1]{References}{ref}
\LastPageEnding

\end{document}